\documentclass[11pt]{article}
\usepackage[a4paper,hmargin=1.0in,vmargin=1.0in]{geometry}
\usepackage{amsmath,amsthm,amssymb,bbm, sectsty,setspace}
\usepackage{graphicx}
\usepackage{multirow, array, arydshln}
\usepackage[round]{natbib}
\usepackage[usenames,dvipsnames]{xcolor}
\usepackage[linktocpage=true,pagebackref=true,colorlinks,linkcolor=BrickRed,citecolor=blue,bookmarks,bookmarksopen,bookmarksnumbered]{hyperref}
\usepackage[normalem]{ulem}
\usepackage{thm-restate}
\usepackage[ruled,vlined]{algorithm2e}
\usepackage{algorithmic}

\usepackage{authblk}

\newcommand{\objfunc}{\Phi}
\theoremstyle{plain}
\newtheorem{theorem}{Theorem}
\newtheorem{example}{Example}
\newtheorem{lemma}[theorem]{Lemma}

\newtheorem{corollary}[theorem]{Corollary}

\sectionfont{\large} \subsectionfont{\normalsize}
\allowdisplaybreaks
\title{The Incremental Knapsack Problem \\ with Monotone Submodular All-or-Nothing Profits}

\author[1]{Federico D'Onofrio}
\author[2]{Yuri Faenza}
\author[3]{Lingyi Zhang}
\affil[1]{DIAG, Sapienza University of Rome}
\affil[2]{IEOR Department, Columbia University}
\affil[3]{Uber Technologies}

\begin{document}

\maketitle

\abstract{We study incremental knapsack problems with profits given by a special class of monotone submodular functions, that we dub \emph{all-or-nothing}. We show that these problems are not harder to approximate than a less general class of modular incremental knapsack problems, that have been investigated in the literature. We also show that certain extensions to more general submodular functions are APX-hard.}

\medskip


\section{Introduction}\label{sec:intro}

Incremental knapsack models are discrete, multi-period extensions of the classical maximum knapsack problem. In such models, we are given a set $[n]=\{1,2,\dots,n\}$ of items with nonnegative weights $w_1,\dots, w_n$, and $T$ capacities $0\leq W_1\leq \dots \leq W_T$. At every time $t \in [T]$, we are allowed to insert items in -- but not remove them from -- the knapsack, as long as the total weight of items currently in the knapsack does not exceed the total capacity $W_t$. Formally, a solution is given by a \emph{chain} $(S_1,\dots, S_T)$, i.e., a family $S_1\subseteq S_2\subseteq \dots \subseteq S_T \subseteq [n]$. We say that a chain is \emph{feasible} if, for each $t \in [T]$, we have $w(S_t)\leq W_t$, where for a function $g$ indexed over a set $[n]$ and for $S\subseteq[n]$, we employ the standard notation $g(S)=\sum_{i \in S}g_i$. The goal is to maximize a profit function, the definition of which depends on the specific model under consideration. 

Incremental knapsack models arise in applications where available resources grow in a predictable manner, allowing a decision-maker to plan for an expansion of their portfolio over time. Consider, for instance, an investor whose budget increases over the course of the year, allowing them to enlarge their set of active investments; or a city council, that can build new infrastructures over the years as more money is collected through taxes and other  sources. We refer to~\citet{faenza2020approximation} and the references therein for details on applications of incremental knapsack problems.

To the best of our knowledge, all research in the area has focused so far on modular profits, i.e., on models where we aim at maximizing a function of the form \begin{equation}\label{eq:modular}\sum_{t \in [T]}\sum_{i \in S_t\setminus S_{t-1}} f(i,t),\end{equation} for some $f : [n] \times [T] \rightarrow \mathbb{Z}_+$ (assuming $S_0=\emptyset$). See the Related work paragraph for examples. In~\eqref{eq:modular}, the profit of an item only depends on whether (and at which time) it is inserted in the knapsack, and not on the other items that are concurrently in the knapsack. Such functions cannot therefore model more complex profits depending on set of items. 

In this paper, we initiate the study of incremental knapsack problems with non-modular profit functions, focusing in particular on certain submodular  functions. Submodular profits can model the presence of substitute goods: in the city council example mentioned above, consider two similar venues that can be built in nearby locations -- say, two playgrounds: the joint profit of building them both is less than the sum of the profits of building each one of them alone. From a theoretical standpoint, maximizing a nonnegative monotone submodular profit function subject to a (non-incremental) knapsack constraint is a classical and well-studied problem, and a tight $(1 - \frac{1}{e})$-approximation can be obtained by a combination of guessing and combinatorial greedy techniques, see~\cite{Sviridenko2004}. However, those techniques do not seem to be effective even if we restrict to modular incremental knapsack problems (see~\cite{aouad2020,Faenza2018APF,faenza2020approximation}), so the quest for tools that can  (approximately) solve submodular maximization under incremental knapsack constraints is open. 

In this work, we focus on the \emph{Monotone Submodular All-or-Nothing Incremental Knapsack problem} (\texttt{IK-AoN}), defined as follows. To every item $i \in [n]$ we associate a \emph{profit} $p_i \in \mathbb{N}=\{1,2,\dots\}$, while to every time $t \in [T]$, we associate a \emph{coefficient} $\Delta_t\in \mathbb{Z}_+ =\{0,1,2,\dots,\}$  and a \emph{capacity} $W_t \in \mathbb{N}$, with $W_1\leq W_2 \leq \dots \leq W_T$. We are moreover given an \emph{aggregation function} $\gamma : 2^{[n]}\rightarrow \mathbb{Z}_+$. The goal is to find a feasible chain $(S_1,S_2,\dots,S_T)$ that maximizes the profit function $\sum_{t \in T}\Delta_t \gamma(S_t)$. Hence, for $S\subseteq [n]$ we let $\Delta_t \gamma(S)$ be the profit gained at time $t$ if the set of items contained in the knapsack is exactly $S$. The function $\gamma$ satisfies the \emph{Monotone Submodular All-or-Nothing} property, i.e., $\gamma(\emptyset)=0$ and the following conditions hold:
\begin{enumerate}
    \item (\emph{Monotone Submodularity}) $\gamma$ is a  monotonically non-decreasing submodular function, that is, for $i \in [n]$ and  $S\subseteq T \subseteq [n]$, we have $\gamma(S\cup \{i\})-\gamma(S) \geq \gamma(T\cup \{i\})-\gamma(T)$.
    \item (\emph{All-or-Nothing Contribution}) for each $i \in [n]$ and $S\subseteq [n]$, we have $\gamma(S\cup \{i\}) - \gamma(S)\in \{0,p_i\}$; 
\end{enumerate} Hence, the addition of item $i$ in $S$ at time $t$ either realizes the ``full profit'' of $\Delta_t p_i$, or no profit at all. In practice, such a profit function models a scenario where $i$ is either a perfect substitute of some item in $S$, or it is not a substitute of any of them. In theory, \texttt{IK-AoN} subsumes interesting special cases, as we discuss next. 

\begin{example}[\texttt{IK}]\label{ex:linear}
Consider an \texttt{IK-AoN} instance, under the additional assumption that $\gamma$ is \emph{modular}, i.e., $\gamma(S)=p(S)$ for $S\subseteq [n]$. The resulting problem is known as \emph{Incremental Knapsack} (\texttt{IK}): it is strongly NP-hard (\citet{bienstock2013approximation}) and admits a PTAS (Polynomial-Time Approximation Scheme, see~\citet{aouad2020}).  
\end{example}

\begin{example}[Matroid rank profits]\label{ex:matroid}
Consider an \texttt{IK-AoN} instance, under the additional assumptions that $\Delta_t=1$ for every $t \in [T]$ and $\gamma$ is the rank function of a matroid. Hence, our goal is to find a family of sets $S_1\subseteq \dots \subseteq S_T$ so that $w(S_t)\leq W_t$ for every $t \in [T]$ and the sum of the ranks of $S_1,\dots, S_T$ is maximized. By a variation of the classical proof of optimality of the greedy algorithm to find an independent set of maximum weight in a matroid (see, for example,~\citet[Chapter 8]{cook2011combinatorial}), or by Algorithm~\ref{alg:IIK-AoN} from the present paper, one deduces that the optimal solution can be obtained with the following greedy procedure. Sort the items $[n]$ such that $w_1 \leq \dots \leq w_n$. For increasing values of $t \in [T]$, build $S_t$ as follows. First set $S_t=S_{t-1}$ (with $S_0=\emptyset$). Then, for increasing values of $i \in [n]$, let $S_t = S_t \cup \{i \}$ if $w(S_t \cup \{i \}) \leq W_t$ and $S_t \cup \{i \}$ is independent (in the classical matroid sense).
\end{example}

Example~\ref{ex:linear} and Example~\ref{ex:matroid} show that \texttt{IK-AoN} contains as special cases certain problems with modular profits, as well as problems with more combinatorial profit structures. 

\smallskip

\noindent {\bf Our contributions.} We define a \emph{modularization} of an instance ${\cal I}$ of \texttt{IK-AoN}, any \texttt{IK} instance that can be obtained as follows. Assume ${\cal I}$ has aggregation function $\gamma$. Drop some of the items, assuming (possibly after renaming) that $[n']=\{1,2,\dots,n'\}$ constitutes the set of remaining items; then define the \texttt{IK} instance with the same profits, weights, coefficients, and capacities, and modular aggregation function $\gamma'$ with $\gamma'(S) = p(S)$ for $S\subseteq [n']$. 
For a family ${\cal C}$ of instances of \texttt{IK-AoN}, we call its modularization the family $\overline {\cal C}$ of all modularizations of all instances from ${\cal C}$. 

As our main result we prove that, if we assume oracle access to the aggregation function, any family of \texttt{IK-AoN} instances is essentially not harder than its modularization. More formally, when denoting by $|{\cal I}|$ the input size of an instance ${\cal I}$, we show the following.

\begin{theorem}\label{thm:IK-AoN-approx}
Let ${\cal C}$ be a class of \texttt{IK-AoN} instances, $\alpha \in [0,1]$. Suppose there is an algorithm that, for each instance $\overline{\cal I}\in \overline {\cal C}$, outputs in time $q(|\overline{\cal I}|)$ an $\alpha$-approximated solution to $\overline{\cal I}$. Then there is an algorithm such that, for each instance ${\cal I} \in {\cal C}$ defined over $n$ items and with aggregation function $\gamma$, outputs in time $q(|{\cal I}|) + O(n\log n) + n \cdot \texttt{oracle}$ an $\alpha$-approximated solution to ${\cal I}$, where $\texttt{oracle}$ is the time required for an evaluation of the function $\gamma$.
\end{theorem}

Theorem~\ref{thm:IK-AoN-approx} and the PTAS for \texttt{IK} by~\citet{aouad2020} imply the following.

\begin{corollary}\label{cor:PTAS-IIK-AON}
When $\gamma$ can be evaluated in time polynomial in the input size, \texttt{IK-AoN} has a PTAS.
\end{corollary}

Moreover, if one aims at practical algorithms for \texttt{IK-AoN} that also have good (though suboptimal) theoretical performance guarantees, then Theorem~\ref{thm:IK-AoN-approx} can be employed by using as a subroutine recent algorithms for \texttt{IK} which have been proved to run fast also on instances of very large size, where Gurobi cannot even output a feasible solution or solve a natural LP relaxation, see~\cite{zhang2022incremental}. 

On the flip side, Theorem~\ref{thm:IK-AoN-approx} implies that the (APX-)hard incremental knapsack problems with submodular, nonnegative profits lie outside the class of monotone submodular all-or-nothing functions (note that such APX-hard instances do exist, because of the APX-hardness of submodular function maximization under a knapsack constraint, see~\cite{feige1998threshold}). As our second result, we show that slightly modifying the assumptions on $\gamma$ in the definition of \texttt{IK-AoN} leads to an APX-hard problem. More formally, define \texttt{IK-$\{0,1,2,3\}$} by replacing condition 2 in the definition of \texttt{IK-AoN} with: 
\begin{itemize}
    \item[2'.] (\emph{$\{0,1,2,3\}$-Contribution})
for each $i \in [n]$ and $S\subseteq [n]$, we have $\gamma(S\cup \{i\})-\gamma(S) \in \{0,1,2,3\}$;\end{itemize}
and assuming that $\gamma$ can be evaluated in time polynomial in the input size.
\begin{restatable}{theorem}{mainhard}
    \label{thm:main:hard} \texttt{IK-$\{0,1,2,3\}$} is APX-hard.
\end{restatable}

\paragraph{Additional Notation.} 
For a chain ${\cal S}=(S_1,\dots, S_T)$ and a set $Q \subseteq [n]$, we write ${\cal S}\subseteq Q$ to denote $S_T \subseteq Q$. Moreover, for $i \in S_T$, we let the \emph{insertion time of $i$} (with respect to ${\cal S}$) to be the smallest $t \in \mathbb{N}$ such that $i \in S_t$.

\paragraph{Related work.}
We have already introduced \texttt{IK} and mentioned that it has a PTAS, based on approximate dynamic programming ideas (\cite{aouad2020}). A relevant special case of \texttt{IK} is the \emph{time-Invariant Incremental Knapsack} (\texttt{IIK}), obtained by setting $\Delta_t=1$ for all $t\in [T]$. \texttt{IIK} is also strongly NP-hard (\cite{bienstock2013approximation}). A PTAS for \texttt{IIK} can be obtained by a combination of guessing, disjunctive programming, and LP rounding (\cite{Faenza2018APF}). A more general problem than \texttt{IK} is the \emph{Generalized Incremental Knapsack problem} (\texttt{GIK}), obtained by letting the objective function be as in~\eqref{eq:modular}. \texttt{GIK} admits an $(\frac{1}{2}-\epsilon)$-approximation and a QPTAS (\cite{faenza2020approximation}), based on a reformulation as a sequencing problem, dynamic programming ideas, and the Shmoys-Tardos algorithm for the generalized assignment problem, among other tools. Note that \texttt{GIK} and \texttt{IK-AoN} are incomparable. More work on incremental knapsack problems has appeared in~\cite{della2018approximating,della2019approximating,zhang2022incremental}.

\paragraph{Organization of the paper.} In Section~\ref{sec:intro:hlw}, we give an outline of the proof of Theorem~\ref{thm:IK-AoN-approx}, together with the corresponding algorithm. A special family of sets, called independent (see Section~\ref{sec:intro:hlw} for a definition) that are crucial for the proof of Theorem~\ref{thm:IK-AoN-approx} are investigated in Section~\ref{sec:ind} and Section~\ref{sec:profit-decomp}. Section~\ref{sec:wrap-up} concludes the proof of Theorem~\ref{thm:IK-AoN-approx}. Section~\ref{sec:hardness} presents the APX-hardness proof for~\texttt{IK-$\{0,1,2,3\}$}.

\section{Outline of the proof of Theorem~\ref{thm:IK-AoN-approx}}\label{sec:intro:hlw} 

We now present the main ideas behind the proof of Theorem~\ref{thm:IK-AoN-approx}. Proofs of the lemma introduced in this section and a formal proof of Theorem~\ref{thm:IK-AoN-approx} are given in later sections.

Fix an \texttt{IK-AoN} instance ${\cal I}$, defined as in Section~\ref{sec:intro}. We call $S\subseteq[n]$ \emph{independent} if $\gamma(S)=\sum_{i \in S} p_i$, \emph{dependent} otherwise. The name is inspired by the rank function $\rho$ of a matroid, for which $\rho(S)=|S|$ if and only if $S$ is independent. We can assume that all sets of the form $\{i\}$ for $i \in [n]$ are independent -- else, it is easy to see that by submodularity $\gamma(S)=\gamma(S\setminus \{i\})$ for all $S\subseteq [n]$, and we can consider the problem restricted to $[n]\setminus \{i\}$. 
Independent sets in our setting share with independent sets in the matroid setting classical properties, e.g., independence is preserved under taking subsets.

\begin{restatable}[Monotonicity of independence]{lemma}{indismonotone}\label{lem:ind-is-monotone} 
Let $S'\subseteq S\subseteq [n]$. If $S$ is independent, then $S'$ is independent.
\end{restatable}

We say that a chain ${\cal S}=(S_1,\dots, S_T)$ is \emph{independent} if $S_1,\dots, S_T$ are independent. Using Lemma~\ref{lem:ind-is-monotone}, the latter is equivalent to $S_T$ being independent. As our first step, we show that restricting to independent chains is enough to obtain an optimal solution.

\begin{restatable}[Optimality of independence]{lemma}{indopt} \label{lem:ind-opt}
There is an optimal chain of ${\cal I}$ that is independent. 
\end{restatable}

The previous lemma calls for an investigation of the structure of independent sets. Let $(P^1,\dots, P^k)$ be the \emph{profit partition} of $[n]$. That is, for $\ell \in [k]$, all items in the \emph{profit class} $P^{\ell}$ have profit $p^{\ell}$, with $0 < p^1 < p^2 < \dots < p^k$. We show the following.

\begin{restatable}[Independence of the union of independent slices]{lemma}{eqprofitdependent}\label{lem:union_of_ind}
Let $S\subseteq [n]$. Then $S$ is independent if and only if, for all $\ell \in [k]$, $P^\ell \cap S$ is independent. 
\end{restatable}

Because of the previous result, we next focus on understanding independent sets contained in each $P^\ell$, $\ell \in [k]$. It turns out that within each profit class, independent sets are very structured, as the next lemma shows. For any $\ell \in [k]$, let ${\cal M}_{\ell} \subseteq 2^{P^{\ell}}$ denote the family of independent sets contained in $P^{\ell}$. 

\begin{restatable}[Matroidal structure of independent sets in a profit class]{lemma}{clsingleprofitmatroid}\label{cl:single_profit_matroid}
For all $\ell \in [k]$, $(P^{\ell}, {\cal M}_{\ell})$ is a matroid.
\end{restatable}

The matroidal structure of independent sets in a profit class implies that the classical greedy algorithm for matroids can be employed to find, for each $\ell \in [k]$, an inclusionwise maximal independent set $P_I^\ell$ of $(P^\ell,{\cal M}_\ell)$ of minimum weight. For $\ell \in [k]$, fix one such $P_I^\ell$ if multiple exist. The next lemma shows an important property of chains contained in $P^\ell$. 

\begin{restatable}[Local optimality of chains contained in $P_I^\ell$]{lemma}{lemrestrictitem}\label{lem:restrict-item}
Given an independent chain ${\cal S}\subseteq P^\ell$, there exists an independent chain ${\cal S}'\subseteq P^\ell_I$ such that, for all $t \in [T]$, we have $\gamma(S_t) = \gamma(S'_t)$ and $w(S'_t) \leq w(S_t)$.
\end{restatable}

 By monotonicity of independence (Lemma~\ref{lem:ind-is-monotone}) and the independence of the union of independent slices  (Lemma~\ref{lem:union_of_ind}), any subset of $\cup_{\ell \in [k]}P_I^\ell$ is independent. So the restriction $\overline{\cal I}$ of ${\cal I}$ to items in $\cup_{\ell \in [k]}P_I^\ell$ is a  modularization of ${\cal I}$.  Using the optimality of independence (Lemma~\ref{lem:ind-opt}) and local optimality of chains contained in $P_I^\ell$ (Lemma~\ref{lem:restrict-item}), it is not hard to see that an optimal solution to $\overline {\cal I}$ is also an optimal solution to ${\cal I}$. We can therefore apply to $\overline {\cal I}$ the $\alpha$-approximation algorithm whose existence is guaranteed by the hypothesis of Theorem~\ref{thm:IK-AoN-approx}, and output the resulting solution. Our approach is summarized in Algorithm~\ref{alg:IIK-AoN}.

\begin{algorithm}
\caption{Algorithm for \texttt{IK-AoN}}
\DontPrintSemicolon
\KwIn{An \texttt{IK-AoN} instance with item set $[n]$ and profits $p_1,\dots,p_n$.}
\begin{algorithmic}[1]
\STATE Let $(P^1,\dots, P^k)$ be the profit partition of $[n]$. 
\FOR{$\ell \in [k]$} 
\STATE Compute an inclusionwise maximal independent set of minimum weight $P_I^\ell$ of the matroid $(P^\ell, {\cal M}_\ell)$.
\ENDFOR 
\STATE Run the $\alpha$-approximation algorithm on the \texttt{IK} instance $\overline{\cal I}$ with item set $\cup_{\ell \in [k]} P_I^\ell$, original weights and capacities,  and aggregation function $\gamma'(S)=\sum_{i \in S}p_i$ for $S\subseteq \cup_{\ell=1}^k P^\ell_I$, as to obtain the chain $\overline{\cal S}$.
\STATE Output $\overline{\cal S}$.
\end{algorithmic}\label{alg:IIK-AoN}
\end{algorithm}

\section{Independent sets}\label{sec:ind}

Fix an \texttt{IK-AoN} instance ${\cal I}$, defined as in Section~\ref{sec:intro}. To study independent sets, we first introduce some relevant concepts and properties, mostly extending analogous ones for matroids. 

\paragraph{Cycles, monotonicity.} We call a non-empty set $C \subseteq [n]$ a \emph{cycle} if $C$ is dependent and $C \setminus \{ i \}$ is independent for every $i \in C$. Cycles have the following interesting property.

\begin{lemma} \label{lem:gamma-of-cycle}
Let $C \subseteq [n]$ be a cycle. Then, for each $i \in C$, we have $\gamma(C)=\gamma(C\setminus \{i\})$. \end{lemma}
\begin{proof}
Let $C,i$ be as in the hypothesis. By definition, $C\setminus \{i\}$ is independent, so $\gamma(C\setminus \{i\})=p(C\setminus \{i\})$. 
By definition, $\gamma(C)=\gamma(C\setminus \{i\})$ or $\gamma(C)=\gamma(C\setminus \{i\})+p_i=p(C)$. Since the latter would imply that $C$ is independent and contradict the hypothesis, the former holds.
\end{proof}

The next lemma shows that each dependent set contains a cycle.
 
\begin{lemma} \label{lem:cycle}
Let $S \subseteq [n]$ be dependent. Then there exists $C \subseteq S$ such that $C$ is a cycle.
\end{lemma}

\begin{proof}
Consider the algorithm that, starting from $C = S$, iteratively removes an item $i \in C$ while $C \setminus \{i \}$ is dependent, and then outputs the resulting set $C$.  We first claim that the set $C$ outputted by the procedure above is non-empty and dependent. Indeed, at the beginning of the algorithm, $C$ is dependent. Moreover, an item is only removed if it preserves the property of $C$ being dependent. Since the empty set is clearly independent, the algorithm halts with a non-empty set. By construction, at termination, $C \setminus \{ i \}$ is independent for all $i \in C$, showing that $C$ is a cycle. 
\end{proof}

Lemma~\ref{lem:ind-is-monotone}, restated here for the reader's convenience, shows that the property of being independent is monotone with respect to set inclusion.

\indismonotone*

\begin{proof}
By hypothesis, $\gamma(S)=p(S)$. Assume by contradiction that there exists $S'\subseteq S$, with $S'$ dependent. Take a cycle $C\subseteq S'$, whose existence is guaranteed by Lemma~\ref{lem:cycle}, and let $i \in C$. We have:
$$0= \gamma(C) - \gamma(C\setminus \{i\}) \geq \gamma(S)-\gamma(S\setminus \{i\}) = p_i,$$
where the first equality follows from Lemma~\ref{lem:gamma-of-cycle},  the inequality by submodularity, and the second equality by independence of $S$. Hence, $p_i \leq 0$, a contradiction to $p_i \in \mathbb{N}$. \end{proof}

\paragraph{Restriction to independent chains.} Recall that we say that a chain ${\cal S}=(S_1,\dots, S_T)$ is \emph{independent} if $S_1,\dots, S_T$ are independent. Using Lemma~\ref{lem:ind-is-monotone}, ${\cal S}$ is independent if and only if $S_T$ is independent. As we show next, in \texttt{IK-AoN} we can restrict our attention to independent chains.

\indopt*

\begin{proof}
Let ${\cal S}^* = (S^*_1, \dots, S^*_T)$ denote an optimal chain of ${\cal I}$. For every $i \in S^*_T$, let $t(i)$ be the insertion time of item $i$ with respect to ${\cal S}^*$. We first claim that, for each $i \in S^*_T$, we have without loss of generality that \begin{equation}\label{eq:t(i)} \gamma(S^*_{t(i)}) - \gamma(S^*_{t(i)} \setminus \{i \})=p_i.\end{equation} Indeed, suppose~\eqref{eq:t(i)} does not hold for some $i \in S^*_T$. We claim that  $\overline {\cal S} = (\overline S_1,\overline S_2,\dots,\overline S_T)$ with $\overline S_t = S^*_t \setminus \{i\}$ for $t \in [T]$ is also an optimal chain. Clearly, $\overline S$ is a feasible chain. Moreover, 
$$0=\gamma(S^*_{t(i)}) - \gamma(S^*_{t(i)} \setminus \{i \}) \geq \gamma(S^*_t) - \gamma(S^*_t \setminus \{i \}),$$ for every $t \geq t(i)$, where the equation holds by definition of $\gamma$ and the hypothesis that~\eqref{eq:t(i)} does not hold, and the inequality by submodularity and the definition of chain. Hence, $\gamma(S^*_t) = \gamma(S^*_t \setminus \{i \})$ for all $t \geq t(i)$ by monotonicity of $\gamma$. Since for $t \in [t(i)-1] $ we have $\overline S_t=S^*_t$, the claim follows. We therefore assume that~\eqref{eq:t(i)} holds for all $i \in S_T^*$.

By way of contradiction, suppose ${\cal S}^*$ is not independent. Thus, there exists some $t \in [T]$  and $i \in S^*_t$ such that $\gamma(S^*_t \setminus \{i \}) = \gamma(S^*_t)$. Let $\tau(i)$ be the smallest time $t \in [T]$ such that $i \in S^*_t$ and $\gamma(S^*_t \setminus \{i \})=\gamma(S^*_t)$. By~\eqref{eq:t(i)}, we know that $\tau(i) > t(i) \geq 1$. Furthermore, $S^*_{\tau(i)} \setminus S^*_{\tau(i)-1} \neq \emptyset$, else $$\gamma(S^*_{\tau(i)}) - \gamma(S^*_{\tau(i)} \setminus \{i \}) = \gamma(S^*_{\tau(i)-1}) - \gamma(S^*_{\tau(i)-1} \setminus \{i \}) = 0,$$ contradicting the choice of $\tau(i)$. 

Let therefore $S^*_{\tau(i)} \setminus S^*_{\tau(i)-1}=\{ j_1, \dots, j_h \}\neq \emptyset$. For $\ell \in [h]$, since the insertion time of $j_{\ell}$ is $\tau(i)$, we know by~\eqref{eq:t(i)} that $p_{j_{\ell}} = \gamma(S^*_{\tau(i)}) - \gamma(S^*_{\tau(i)} \setminus \{j_{\ell} \})$. Furthermore, 
$$    \gamma(S^*_{\tau(i)}) - \gamma(S^*_{\tau(i)} \setminus \{j_{\ell} \}) \leq \gamma(S^*_{\tau(i)-1} \cup \{j_1, \dots, j_{\ell} \}) - \gamma(S^*_{\tau(i)-1} \cup \{j_1, \dots, j_{\ell-1}\})\leq p_{j_{\ell}},$$ where the first inequality holds by submodularity (since $S^*_{\tau(i) - 1} \cup \{j_1, \dots, j_{\ell-1} \} \subseteq S^*_{\tau(i)} \setminus \{j_{\ell}\}$) and the second by the all-or-nothing property. Hence, 
\begin{equation} \label{eq:profit_1}
p_{j_{\ell}} = \gamma(S^*_{\tau(i)-1} \cup \{j_1, \dots, j_{\ell} \}) - \gamma(S^*_{\tau(i)-1} \cup \{j_1, \dots, j_{\ell-1}\}).
\end{equation}
Since $\tau(i)>t(i)$, we know $i \neq j_\ell$ for all $\ell \in [h]$. Hence, by employing a reasoning similar to the one that led to~\eqref{eq:profit_1}, we have:
\begin{equation} \label{eq:profit_2}
    p_{j_{\ell}} = \gamma(S^*_{\tau(i)-1} \setminus \{i \} \cup \{j_1, \dots ,j_{\ell} \}) - \gamma(S^*_{\tau(i)-1} \setminus \{i \} \cup \{j_1, \dots, j_{\ell-1}\}).
\end{equation}
We have therefore
\begin{equation}\label{eq:eq-1}\sum_{\ell = 1}^h p_{j_{\ell}} = \gamma(S^*_{\tau(i) -1} \cup \{j_1, \dots, j_h\}) - \gamma(S^*_{\tau(i)-1}) =  \gamma(S^*_{\tau(i)}) - \gamma(S^*_{\tau(i)-1}),
\end{equation}
where the first equality follows by  repeatedly replacing $p_{j_\ell}$ with the right-hand side of Equation~\eqref{eq:profit_1}, and then by telescoping cancellations, and the second equality follows by definition.

Similarly, summing~\eqref{eq:profit_2} for all $\ell \in [h]$, we have
\begin{equation}\label{eq:eq-2}\sum_{\ell = 1}^h p_{j_{\ell}} = \gamma(S^*_{\tau(i) -1} \setminus \{i \} \cup \{j_1, \dots, j_h\}) - \gamma(S^*_{\tau(i)-1}\setminus \{i \}) = \gamma(S^*_{\tau(i)}\setminus \{i \}) - \gamma(S^*_{\tau(i)-1}\setminus \{i \}).\end{equation}
Combining~\eqref{eq:eq-1} and~\eqref{eq:eq-2}, we have
$$\gamma(S^*_{\tau(i)}) - \gamma(S^*_{\tau(i)-1})  =  \gamma(S^*_{\tau(i)}\setminus \{i \}) - \gamma(S^*_{\tau(i)-1}\setminus \{i \})$$
and therefore, using the assumption that $0=\gamma(S^*_{\tau(i)}) - \gamma(S^*_{\tau(i)} \setminus \{i \})$, we deduce
$
\gamma(S^*_{\tau(i) - 1}) =  \gamma(S^*_{\tau(i) -1} \setminus \{i \}).
$ We have reached a contradiction since we assumed that $\tau(i)$ is the smallest time $t \in [T]$ such that $\gamma(S^*_{t} \setminus \{i \})=\gamma(S^*_{t})$.
\end{proof}

\paragraph{Union of independent slices.}

Recall that the \emph{profit partition} $(P^1,\dots, P^k)$ of $[n]$ is such that, for $\ell \in [k]$, all items in the \emph{profit class} $P^{\ell}$ have profit $p^{\ell}$, with $0 < p^1 < p^2 < \dots < p^k$. The next lemma shows that the independence of $S\subseteq [n]$ is completely determined by the independence of the restrictions of $S$ to each of the profit classes.

\eqprofitdependent*

\begin{proof}
The necessity condition follows immediately from Lemma~\ref{lem:ind-is-monotone}. 
For the sufficiency condition, we show the contrapositive. Thus, let $S$ be dependent. We show that there exists a dependent subset of $S$ whose elements all have the same profit. The proof is by induction on $|S|$. Let $|S| = 1$. Since we assume all sets of cardinality $1$ are independent (see Section~\ref{sec:intro:hlw}), the statement is vacuously true. Hence, let $|S|\geq 2$ and $C \subseteq S$ be a cycle, whose existence is guaranteed by Lemma~\ref{lem:cycle}. If $|C| < |S|$, then by inductive hypothesis, there exists $S' \subseteq C \subseteq S$ such that $S'$ is dependent and contains only items with equal profits.
Else, $S$ is a cycle. Hence, for $i \in S$, we have $$\gamma(S)=\gamma(S \setminus \{i \})=\sum_{j \in S \setminus \{i\}} p_j,$$ where the first equation holds by Lemma~\ref{lem:gamma-of-cycle} and the second by independence. Thus, for any $i,j \in S$,
$$\sum_{r\in S \setminus \{i \}} p_{r} = \gamma(S)=  \sum_{r \in S \setminus \{j \}} p_r.$$
Cancelling out all $p_r$ for $r \notin \{i, j \}$ in the equality above, we get $p_i = p_j$. Since $i,j$ were chosen arbitrarily from $S$, all items from $S$ have the same profit, concluding the proof. \end{proof}

\section{Independent sets in single profit classes}\label{sec:profit-decomp}

Fix again an \texttt{IK-AoN} instance ${\cal I}$, defined as in Section~\ref{sec:intro}, and let $(P^1,\dots, P^k)$ be the profit partition of $[n]$. 

\paragraph{Slicing by profit.} 

The next lemma gives an important property of independent sets contained in a single profit class. Recall that, for any $\ell \in [k]$, we let ${\cal M}_{\ell} \subseteq 2^{P^{\ell}}$ denote the family of independent sets of $P^{\ell}$. 

\clsingleprofitmatroid*

\begin{proof}
Fix $\ell \in [k]$. Trivially, $\emptyset \in {\cal M}_{\ell}$. For any set $S \in {\cal M}_{\ell}$ and any $S' \subseteq S$, if $S$ is independent, so is $S'$ by Lemma~\ref{lem:ind-is-monotone}. To conclude that $(P^\ell,{\cal M}_\ell)$ is a matroid, pick $A \subseteq P^{\ell}$ and two inclusionwise maximal independent sets $S, S' \subseteq A$. It suffices to show that $|S| = |S'|$ (see, e.g.,~\cite[Exercise 3 on page 14]{oxley2006matroid}). By way of contradiction, without loss of generality, assume $|S| > |S'|$. By monotonicity, and since we know that both $S$ and $S'$ are independent sets whose all items have profit $p^\ell$, we have
\begin{equation}\label{eq:matroid}\gamma(S \cup S') \geq  \gamma(S) = |S|\cdot p^{\ell} > |S'| \cdot p^{\ell}= \gamma(S').\end{equation}
Let $S \setminus S'=\{j_1, \dots, j_h \}$. We claim that there exists $\ell\in [h]$ such that $S' \cup \{j_\ell\}$ is independent, a contradiction to the choice of $S'$. Indeed, if for each $\ell \in [h]$ we have that $S'\cup \{j_\ell\}$ is dependent, then $\gamma(S'\cup \{j_\ell\})=\gamma(S')$ by definition. Hence, by submodularity, 
$$
0 = \gamma(S'\cup \{j_\ell\})- \gamma(S') \geq \gamma(S' \cup \{j_1,\dots,j_\ell\}) - \gamma(S' \cup \{j_1,\dots,j_{\ell-1}\}),
$$
Thus, by telescoping sum
$$
0 \geq \gamma(S' \cup \{j_1,\dots,j_h\})-\gamma(S')=\gamma(S' \cup S) - \gamma(S'),
$$
contradicting~\eqref{eq:matroid}. \end{proof} 

\paragraph{Restriction to independent sets of minimum weight.} 

Fix $\ell \in [k]$. Let $P_I^\ell$ be an inclusionwise maximal independent set of minimum weight from the matroid $(P^\ell,{\cal M}_\ell)$. 
Note that since $P_I^{\ell}$ is an independent set, all subsets of $P_I^{\ell}$ are also independent by Lemma~\ref{lem:ind-is-monotone}.

\lemrestrictitem*

\begin{proof}
We construct ${\cal S'}=(S_1',S_2',\dots,S_T')$ as follows. For every $t \in [T]$, let $S'_t$ be the $|S_t|$ items of minimum weight in $P_I^{\ell}$, breaking ties by choosing smaller index items. To show that ${\cal S'}$ is well-defined, observe that $S_T\subseteq P^\ell$ is independent by hypothesis, and $P_I^{\ell}$ is by definition inclusionwise maximal among independent sets contained in $P^\ell$. By basic matroid properties, we have $|S_T| \leq |P_I^{\ell}|$. 

Since ${\cal S}$ is a chain, ${\cal S}'$ is also a chain.  ${\cal S}$ is independent by hypothesis, and ${\cal S}'$ is independent by Lemma~\ref{lem:ind-is-monotone} since ${\cal S}'\subseteq P_I^\ell$ and $P_I^\ell$ is independent. Since all items in ${P}^{\ell}$ have equal profit, we have $\gamma(S_t) = \gamma(S'_t)$ for all $t \in [T]$.

We are left to show $w(S_t')\leq w(S_t)$ for all $t \in [T]$. Suppose by contradiction there exists $t \in [T]$ such that $w(S'_t) > w(S_t)$. Let $S'_t = \{j'_1, \dots, j'_m \}$ such that $w_{j'_1} \leq \dots \leq w_{j'_m}$. Similarly, let $S_t = \{j_1, \dots, j_m \}$ be such that $w_{j_1} \leq \dots \leq w_{j_m}$ (recall that $|S'_t|=|S_t|$ by construction). Let $h$ be the smallest index such that $w_{j'_h} > w_{j_h}$. The existence of $j'_h$ follows from $w(S'_t) > w(S_t)$ and $|S_t|=|S'_t|$. Let $J=\{ j_1, \dots j_h \}$, $J'=\{j'_1, \dots, j'_{h-1} \}$, $A = J \cup J'$.

  We claim that $J'$ is an inclusionwise maximal independent set contained in $A$. Since $J'\subseteq S'_t$ and we argued above that ${\cal S}'$ is independent, we can use Lemma~\ref{lem:ind-is-monotone} to conclude that $J'$ is independent. Let us now argue inclusionwise maximality. By definition, $w_j \leq w_{j_h}< w_{j'_h}$ for all $j \in J$. Fix $j \in J \setminus J'$. Since $j'_h \in S'_t$ and $j \notin J'$, by definition of $S'_t$ we have that $j  \notin P_I^{\ell}$. Consider therefore the step in the construction of $P^{\ell}_I$ when $j'_h$ is added to $P^\ell_I$. Since $j$ is not added, we have that $J' \cup \{j \}$ is dependent. Hence, $J'$ is an inclusionwise maximal independent set of $A$. 
  
  Now, $J\subseteq S_t$ is independent by hypothesis and Lemma~\ref{lem:ind-is-monotone}. Moreover, $|J|>|J'|$, a contradiction since $(P^{\ell}, {\cal M}_{\ell})$ forms a matroid by Lemma~\ref{cl:single_profit_matroid}. Hence, it must be that $w(S'_t) \leq w(S_t)$ for all $t \in [T]$, concluding the proof.
\end{proof}

\section{Proof of Theorem~\ref{thm:IK-AoN-approx}}\label{sec:wrap-up}

\paragraph{Well-definedness of Algorithm~\ref{alg:IIK-AoN}.} Let ${\cal I}$, $\overline{\cal I}$ be as in Algorithm~\ref{alg:IIK-AoN}. Notice that $\overline {\cal I}$ is a modularization of ${\cal I}$. Hence, we can apply the $\alpha$-approximation algorithm whose existence is guaranteed by the hypothesis of Theorem~\ref{thm:IK-AoN-approx}. Note that $\overline{\cal S}$, as output by the algorithm, is a feasible chain of both $\overline{\cal I}$ (by construction) and ${\cal I}$ (since modularization does not affect feasibility). Hence, Algorithm~\ref{alg:IIK-AoN} outputs a feasible solution to ${\cal I}$.

\paragraph{Approximation guarantee of Algorithm~\ref{alg:IIK-AoN}.} Given any chain ${\cal S}=(S_1,\dots,S_T)\subseteq [n]$, let $\objfunc({\cal S}) = \sum_{t \in T} \Delta_t \cdot \gamma(S_t)$ denote its profit in the instance ${\cal I}$.  If moreover ${\cal S}\subseteq \cup_{\ell \in [k]}P^\ell_I$, let $\overline\objfunc({\cal S})$ denote the profit chain ${\cal S}$ earns in the instance $\overline {\cal I}$. Since $\overline {\cal I}$ is an \texttt{IK} instance, $\overline \objfunc({\cal S}) = \sum_{t \in T} \Delta_t \sum_{i \in S_t} p_i$. 
The fact that the chain $\overline {\cal S}$ output by the algorithm is an $\alpha$-approximated solution to ${\cal I}$ immediately follows from the next lemma. 

\begin{lemma} \label{lem:ik-reduction}
Let ${\cal S}$ be an optimal solution of ${\cal I}$ and let ${\cal S}'$ be an optimal solution of $\overline {\cal I}$. Then $\overline\objfunc({\cal S}') \geq \objfunc({\cal S})$. Furthermore, given any solution $\hat {\cal S}$ feasible for $\overline{\cal I}$, then $\hat {\cal S}$ is feasible for ${\cal I}$ and $\objfunc(\hat {\cal S}) = \overline \objfunc(\hat {\cal S})$.
\end{lemma}
\begin{proof}
Let us start with the first statement. By Lemma~\ref{lem:ind-opt}, we can assume without loss of generality that ${\cal S}$ is independent. We decompose ${\cal S}=(S_1,\dots,S_T)$ into $k$ separate chains, with each chain containing only items in $P^{\ell}$ for some $\ell \in [k]$. More precisely, for each $\ell \in [k]$, let ${\cal S}^\ell = (S_1 \cap P^{\ell}, \dots, S_T \cap P^{\ell}).$ One easily verifies that ${\cal S}^\ell$ is indeed a chain for each $\ell \in [k]$. Moreover, since ${\cal S}$ is independent, ${\cal S}^{\ell}$ is also independent by Lemma~\ref{lem:union_of_ind}.
It follows that
\begin{equation} \label{eq:opt-profit}
\objfunc({\cal S}) = \sum_{t \in [T]} \Delta_t \cdot \gamma(S_t) = \sum_{t \in [T]} \Delta_t \cdot p(S_t) =  \sum_{t \in [T]} \Delta_t \cdot \sum_{\ell \in [k]} p(S^{\ell}_t)=  \sum_{t \in [T]} \Delta_t \cdot \sum_{\ell \in [k]} \gamma(S^{\ell}_t).
\end{equation}
By Lemma~\ref{lem:restrict-item}, for each $\ell \in [k]$ there exists an independent chain of ${\cal I}$, call it $\tilde {\cal S}^{\ell}=(\tilde S^\ell_1,\dots, \tilde S^\ell_T)$, such that $\tilde {\cal S}^{\ell}\subseteq P_I^{\ell}$, $\gamma(\tilde{S}^{\ell}_t) = \gamma(S^{\ell}_t)$, and $w(\tilde{S}^{\ell}_t) \leq w(S^{\ell}_t)$ for $t \in [T]$.

Let ${\tilde {\cal S}}=(\tilde S_1,\dots,\tilde S_T)$ be defined as $\tilde S_t = \cup_{\ell \in [k]} \tilde S^\ell_t$ for $t \in [T]$. Since, for $\ell \in [k]$, we have that $\tilde {\cal S}^\ell$ is a chain of ${\cal I}$ contained in $P_I^\ell$, then $\tilde {\cal S}$ is a chain of $\overline {\cal I}$. To see that ${\tilde {\cal S}}$ is feasible for $\overline {\cal I}$, note that for each $t \in [T]$,
$$w(\tilde {S}_t) = \sum_{\ell \in [k]} w(\tilde {S}^\ell_t) \leq \sum_{\ell \in [k]} w(S^\ell_t) = w(S_t) \leq W_t,$$
where the first inequality follows by Lemma~\ref{lem:restrict-item}, and the final inequality follows by feasibility of ${\cal S}$ in ${\cal I}$. Moreover, $\tilde {\cal S}$ is independent in ${\cal I}$  by Lemma~\ref{lem:union_of_ind}. Hence,
\begin{eqnarray*}
\overline\objfunc(\tilde{\cal S}) & = & \sum_{t \in [T]} \Delta_t \sum_{\ell \in [k]} \sum_{i \in {\tilde S}_t \cap P_I^{\ell}} p_i \\
& = & \sum_{t \in [T]} \Delta_t \sum_{\ell \in [k]} \gamma(\tilde{S}^{\ell}_t) \\
& = & \sum_{t \in [T]} \Delta_t \sum_{\ell \in [k]}\gamma(S^{\ell}_t) \\
& = & \objfunc({\cal S}),
\end{eqnarray*}
where the first equality follows by definition, the second since, as argued above, $\tilde {\cal S}$ is independent in ${\cal I}$, the third by Lemma~\ref{lem:restrict-item}, and the final equality by~\eqref{eq:opt-profit}.
Following the analysis above, $\tilde {\cal S}$ is feasible in $\overline {\cal I}$ with $\overline \objfunc(\tilde{\cal S}) = \objfunc({\cal S})$. It follows that the optimal chain ${\cal S}'$ of $\overline{\cal I}$ satisfies $\overline\objfunc({\cal S}') \geq \objfunc({\cal S})$. This concludes the proof of the first part of the lemma.

To prove the second part of the lemma, let $\hat {\cal S}$ be a feasible solution of $\overline{\cal I}$. Then it is clearly feasible for ${\cal I}$. To show that $\objfunc(\hat {\cal S}) = \overline\objfunc(\hat{\cal S})$, it suffices to show that $\hat {\cal S}$ is an independent chain in ${\cal I}$. We again decompose $\hat {\cal S}$ by profit into $k$ separate chains $\hat {\cal S}^1,\dots, \hat{\cal S}^k$ as done in the first part of the proof. 
For $\ell \in [k]$, using the fact that $P_I^{\ell}$ is independent by construction, and Lemma~\ref{lem:union_of_ind}, we have that $\hat {\cal S}^\ell$ is independent. Using again Lemma~\ref{lem:union_of_ind}, $\hat {\cal S}$ is independent in ${\cal I}$. \end{proof}

\paragraph{Analysis of the running time.} $P^1,\dots, P^k$ can be created by ordering the items by profit, which takes time $\Theta(n \log n)$. $P^1_I,\dots, P^k_I$ can be constructed by $O(n)$ calls to the evaluation oracle for $\gamma$. Running the $\alpha$-approximation algorithm on $\overline {\cal I}$ takes at most $q(|{\cal I}|)$ time.

\section{Which submodular incremental knapsack problems are hard to approximate?} \label{sec:hardness}

In this section, we prove Theorem~\ref{thm:main:hard}, that we restate here for convenience.

\mainhard*

\begin{proof}
We show that the problem is APX-hard through a reduction from the max $k$-vertex cover on subcubic graphs. 

In the max $k$-vertex cover on a subcubic graph problem, we are given a subcubic graph $G=(V,E)$, i.e., such that each $v \in V $ has degree $d(v)\leq 3$, and a positive integer $k$.   The goal is to find a subset $V' \subseteq V$ where $|V'| \leq k$ such that the number of edges that $V'$ covers is maximized. Unless P = NP, there does not exist a PTAS for the max $k$-vertex cover on subcubic graphs, see~\cite{petrank1994hardness}.

Let $V=\{v_1,\dots,v_n\}$. For $v_i \in V$, we create an item $i \in [n]$, with $w_i = 1$. Additionally, let $T = 1$, $W_1 = k$ and $\Delta_1 = 1$. For every $S \subseteq [n]$, let $v(S) = \{v_i \in V: i \in S\}$. Let $\gamma(S)$ denote the number of edges $v(S)$ covers. Notice that $\gamma(\emptyset)=0$, and $\gamma$ can clearly be evaluated in time polynomial in the input size. We show that $\gamma$ is a monotone, submodular function that satisfies the $\{0,1,2,3\}$-Contribution property: 
\begin{itemize}
\item \emph{Monotone Submodularity}:
For every $S \subseteq T \subseteq [n]$ and every $i \in [n]$, we have $\gamma(S \cup \{ i \}) - \gamma(S) \geq \gamma(T \cup \{ i \}) - \gamma(T)$, since $v(T)$ covers every edge that $v(S)$ covers.  
\item \emph{$\{0,1,2,3\}$-Contribution}: By monotonicity, for $i \in [n]$ and $S\subseteq [n]$, $\gamma(S \cup \{i\}) - \gamma(S)$ is a nonnegative integer. Since $G$ is a subcubic graph, $\gamma(S \cup \{i\}) - \gamma(S) \leq d(v_i) \leq 3$. 
\end{itemize}

It is easy to see that for every solution $S\subseteq \{v_1,\dots,v_n\}$ to the \texttt{IK}-$\{0,1,2,3\}$ instance above, $v(S)$ is a feasible solution for the max $k$-vertex cover covering exactly $\gamma(S)$ edges. Vice-versa, every $k$-vertex cover $v(S)$ covering $q$ edges corresponds to a feasible solution $S$ for the \texttt{IK}-$\{0,1,2,3\}$ instance above with $\gamma(S)=q$.

Suppose by contradiction that there is a PTAS solving the \texttt{IK}-$\{0,1,2,3\}$ instances created as above from the $\max$ $k$-vertex cover instances.
Let $S^*$ be the optimal solution to the problem and let $\bar{S}$ be the $(1- \epsilon)$-approximation of $S^*$ as obtained by the PTAS. Hence, $v(S^*)$ is an optimal solution to the max $k$-vertex cover instance, and $v(\bar{S})$ covers at least an $(1 - \epsilon)$ fraction of the vertices covered by $v(S^*)$. Since $\bar{S}$ is obtainable in polynomial time, so is $v(\bar{S})$. Thus, we have obtained a PTAS for the max $k$-vertex cover on a cubic graph problem, deducing the required contradiction and concluding the proof. \end{proof}

As the reader has probably already remarked, the instance created in the proof of Theorem~\ref{thm:main:hard} is a classical submodular function maximization problem under knapsack constraint, i.e., \texttt{IK}-$\{0,1,2,3\}$ is APX-hard already when $T=1$.

\medskip 

\noindent {\bf Acknowledgments.} Yuri Faenza and Lingyi Zhang acknowledge support by the ONR grant N00014-20-1-2091. Most of this work was conducted when Lingyi Zhang was a Ph.D.~student, and Federico D'Onofrio a visiting student, at the IEOR department, Columbia University.

\bibliographystyle{plainnat}
\bibliography{IKmain}
\end{document}